\documentclass[11pt, twoside, letterpaper]{article}
\usepackage[margin = 1in]{geometry}
\usepackage{comment}
\usepackage{indentfirst}
\usepackage[pdftex]{graphicx}
\usepackage{amsmath}
\usepackage{amssymb}
\usepackage{framed}
\usepackage{aliascnt}
\usepackage{hyperref}
\usepackage{multirow, multicol}
\usepackage{float, graphicx}
\usepackage{verbatim, rotating, paralist}
\usepackage{caption}
\usepackage{subcaption}
\usepackage{pstricks, sgamevar, egameps}
\usepackage[framed, hyperref, thmmarks, amsmath]{ntheorem}
\theoremseparator{.}

\newcommand{\Real}{\mbox{${\mathbb R}$}}

\newcommand{\sss}{\mbox{\boldmath $s$}}
\newcommand{\xx}{\mbox{\boldmath $x$}}
\newcommand{\yy}{\mbox{\boldmath $y$}}
\newcommand{\ssigma}{\mbox{\boldmath $\sigma$}}

\newcommand{\zz}{\mbox{\boldmath $z$}}

\newcommand{\D}{\mbox{$\Delta$}}

\newcommand{\defeq}{\stackrel{\textup{def}}{=}}
\DeclareMathOperator*{\argmax}{arg\,max}

\makeatletter
\newcommand{\mypm}{\mathbin{\mathpalette\@mypm\relax}}
\newcommand{\@mypm}[2]{\ooalign{%
  \raisebox{.1\height}{$#1+$}\cr
  \smash{\raisebox{-.6\height}{$#1-$}}\cr}}
\makeatother

\newtheorem{thm}{Theorem}
\newaliascnt{lem}{thm}
\newtheorem{lem}[lem]{Lemma}
\newtheorem{definition}[thm]{Definition}

\aliascntresetthe{lem}

\newaliascnt{mydef}{thm}

\aliascntresetthe{mydef}

\newaliascnt{prop}{thm}
\newtheorem{prop}[prop]{Proposition}
\aliascntresetthe{prop}

\newaliascnt{clo}{thm}
\newtheorem{clo}[clo]{Corollary}
\aliascntresetthe{clo}

\theoremnumbering{Alph}
\newframedtheorem{algo}{Algorithm}

\theoremnumbering{arabic}

\theorembodyfont{\upshape}
\newframedtheorem{algoa}{Algorithm}

\theoremseparator{}
\theoremsymbol{\rule{1ex}{1ex}}
\theoremstyle{nonumberplain}
\newtheorem{proof}{Proof}

\newcommand{\etal}{{\em et al.~}}

\DeclareMathOperator*{\argmin}{arg\,min}

\def\E{\ensuremath{\mathbf{E}}}

\usepackage{titling}
\newcommand{\subtitle}[1]{%
  \posttitle{%
    \par\end{center}
    \begin{center}\large#1\end{center}
    \vskip0.5em}%
}
\newcommand{\HL}[1]{\textcolor{blue}{#1}}
\newcommand{\suppress}[1]{}
\newcommand{\be}{\begin{equation}}
\newcommand{\ee}{\end{equation}}
\newcommand{\bea}{\begin{eqnarray}}
\newcommand{\eea}{\end{eqnarray}}
\newcommand{\bean}{\begin{eqnarray*}}
\newcommand{\eean}{\end{eqnarray*}}

\begin{document}

\title{Learning Dynamics and the Co-Evolution of Competing Sexual Species}

\author{
Georgios Piliouras\thanks{Supported in part by
 SUTD grant SRG ESD 2015 097, MOE AcRF Tier 2 Grant  2016-T2-1-170 and a NRF Fellowship. Part of the work was completed while GP was a CMI Wally Baer and Jeri Weiss postdoctoral  scholar at Caltech. Part of the work was completed while GP was a Simons Institute research fellow.}\\Singapore University of Technology and Design\\georgios@sutd.edu.sg
\and 
Leonard J. Schulman\thanks{Supported in part by NSF grants 1319745 and 1618795, and, while in residence at the Israel Institute for Advanced Studies, by 
a EURIAS Senior Fellowship co-funded by the Marie Sk{\l}odowska-Curie Actions under the 7th Framework Programme.}\\California Institute of Technology\\schulman@caltech.edu  
}

\date{}

\maketitle

\begin{abstract}
We analyze a stylized model of co-evolution between any two purely competing species (e.g., host and parasite), both sexually reproducing. 
Similarly to a recent model of Livnat \etal~\cite{evolfocs14} the fitness of an individual depends on whether the truth assignments on $n$ variables that reproduce through recombination satisfy a particular Boolean function.
Whereas in the original model a satisfying assignment always confers a small evolutionary advantage, in our model the two species are in an evolutionary race with the parasite enjoying the advantage if the value of its Boolean function
matches its host, and the host wishing to mismatch its parasite. Surprisingly, this model makes a simple and robust behavioral  prediction. The typical system behavior is \textit{periodic}. These cycles stay bounded away from the boundary and thus, \textit{learning-dynamics competition between sexual species can provide an explanation for genetic diversity.} This explanation is due solely to the natural selection process. No mutations, environmental changes, etc., need be invoked. 

The game played at the gene level may have many Nash equilibria with widely diverse fitness levels. Nevertheless, sexual evolution leads to gene coordination that implements an optimal strategy, i.e., an optimal population mixture, at the species level. Namely, the play of the many ``selfish genes'' implements a time-averaged correlated equilibrium where the average fitness of each species is exactly equal to its value in the two species zero-sum competition. 

Our analysis combines tools from game theory, dynamical systems and Boolean functions to establish a novel class of  conservative dynamical systems. 

 \end{abstract}


\thispagestyle{empty}
\newpage
\setcounter{page}{1}

\pagenumbering{arabic}




\section{Introduction}

\bigskip

An exciting recent line of work in the theory of computation has focused on the algorithmic power of the evolutionary process (Valiant~\cite{Valiant-book}, Livnat et al.~\cite{evolfocs14,CACM}). The latter two papers identified as interesting the case of a sexually reproducing, haploidal, and panmictic (explained below) species, evolving in a fixed environment according to variants of Multiplicative Weights Update dynamics~\cite{Chastain:2013:MUC:2422436.2422444,PNAS2:Chastain16062014}---which are typically referred to as ``replicator dynamics'' in the evolutionary dynamics literature~\cite{Weibull}. Curiously, however, 
Mehta et al.~\cite{ITCS15MPP} made the discovery that these dynamics lead in the long run (in almost all cases) to a genetic monoculture. This rather contradicts the evidence of natural diversity around us. 


Several plausible explanations exist for this discrepancy, including:
(a) mutations~\cite{ITCS17MPPPV}, (b) speciation (\textit{e.g.}, the Bateson-Dobzhansky-Muller model)~\cite{mathbio1}, (c) the mathematical assumptions are too far from reality, (d) ``in the long run" is longer than geologic time.
%
There is, however, a long-standing argument,
 that there is another (and perhaps more important) factor driving diversity; 
to our knowledge this case was first compellingly laid out by Ehrlich and Raven in 1964~\cite{EhrlichR64}: ``It is apparent that reciprocal selective responses have been greatly underrated as a factor in the origination of organic diversity."
(Already Darwin noted the significance of co-evolution, e.g., between orchids and moths that feed on their nectar; but the proposed implication for diversity seems to have come later.)
\suppress{
Already Darwin noted that a major component of a species' environment  is not exogenous to the evolutionary process (conditions such as sunshine, rainfall etc.)---but
is its interactions with other species, and that these \textit{are themselves undergoing evolution,} that is to say co-evolution affected by the same interactions.
}
In the ensuing decades this idea played a role in the \textit{Red Queen Hypothesis}~\cite{valen73} and was advanced as an explanation of an advantage of sexual over asexual reproduction~\cite{Bell82}.

Apart from empirical study (e.g.,~\cite{brodie2002the,ThompsonC02,10.2307/3298524,doi:10.1086/376580}), 
the dynamics of co-evolution have also been studied mathematically, but primarily (explicitly or implicitly) for asexual reproduction---dynamics in which the abundance of a genome changes over time in proportion to its fitness (possibly with mutations), as in the work of Eigen, Schuster and others~\cite{Eigen71,EigenS79,NowakO08,Thompson94,Thompson05}. The case of sexual reproduction, however, is quite different. There is a good mathematical model for these dynamics, called the 
``weak selection" model~\cite{Nagylaki1}, but effects of co-evolution are not yet understood in this model.

We study a specific class of systems in this model, and provide  
a quantitative study of the evolutionary dynamics of sexual species in highly competitive (``zero sum") interactions. This study supports the thesis of Ehrlich and Raven, that competition drives diversity, in a strong form: not only does a genetic monoculture not take over, but in fact the entropy of the species' genomes is bounded away from $0$ for all time. 
Thus we support a rationale for ecosystem diversity without invoking 
mutation, speciation or environmental change. 

%
 \suppress{
As we show in this paper, the co-evolution of competing sexual species
can 
drive dynamics which  (from most starting conditions) do not equilibrate, but instead oscillate. Moreover, these dynamics  guarantee positive entropy in the genotypes for all time. Both of these effects are contrary to the above-referenced ``monoculture" phenomenon for evolution of a single sexual species against a fixed environment.
}

A sexual species under weak evolutionary pressures (to be made precise) can be modeled, game-theoretically, as a \textit{team,} whose players are the genes \cite{evolfocs14,PNAS2:Chastain16062014}. A team in a multiplayer game~\cite{Marschak55,SchV17} is a set of players who share a common payoff but use independent randomness.
In our setting we have two teams that compete against each other for survival.
 Learning dynamics in such games create dynamical systems entirely unlike those explored in the no-regret learning in games literature~\cite{pota}. For example although there have been quite a few papers arguing about nonequilibrium limit cycles in game dynamics 
 \cite{NonConvergingDask,paperics11,sigecom11,CRS16,ChaosNIPS17}, these typically explore small games with a maximum of two or three players, each with two or three available strategies. Even in these settings the analysis is typically intricate and is based on case-by-case observations that do not generalize to large classes of games. In a handful of cases where non-equilibrium behavior is proven \cite{Soda14,PiliourasAAMAS2014,CyclesSODA18} in larger games, the non-equilibrium behavior is typically non-periodic, and thus to a large extent unpredictable. In contrast, we present a large parametric class of multi-agent games for which we prove periodicity from almost all initial conditions. 
 
{\bf Class of Games:} We establish our results in the restrictive setting of ``Boolean phenotypes". To explain, we will be studying a zero-sum competition between two species $A$ and $B$. Organism $A$ has $n$ genes, and organism $B$ has $m$ genes. Each is haploid (possessing, for each gene, exactly one of the possible alleles), so that a genotype of species $A$ is a vector $s=(s_1,\ldots,s_n)$, and that of species $B$ is $\sigma=(\sigma_1,\ldots,\sigma_m)$. In the Boolean-phenotype model there are Boolean-valued functions and a $2\times 2$ payoff matrix $U$ such that the result of an encounter between $s$ and $\sigma$ is a payoff of $U(f(s),g(\sigma))$ (to $A$, and minus this to $B$).\footnote{Like any multiplayer game, such team zero-sum games have Nash equilibria; but unlike the two-\textit{player} zero-sum game which they resemble, the team game generally has a positive duality gap. (See~\cite{SchV17} where this is worked out, and earlier~\cite{VONSTENGEL1997309} for the case of a team vs.\ a single player.) This gap creates an opportunity for very rich dynamics when each gene continually adjusts its allele frequencies to the competition.} 
As compared with the full reach of evolutionary dynamics this is a limited setting, but
 there are two good reasons to examine it.

\begin{enumerate}
\item This already constitutes an extension to \textit{two} competing and adapting learners, of 
the study of evolutionary learning of Boolean functions initiated in~\cite{evolfocs14}.
\item There are biological examples which approximately fit this assumption---not with respect to the entire phenotype, but w.r.t.\ that aspect of the phenotype which is critical to the two-species interaction. E.g., the length of a hummingbird's beak vs.\ that of a flower; the size of a crab's pincers vs.\ the shell thickness of its prey; the choice of protein coatings employed by a microbe, and the corresponding immune response. 
\end{enumerate}

Not every $2 \times 2$ game is, of course, zero sum; \textit{e.g.}, above, the interaction may be favorable to the hummingbird and to the flowering plant. However, we are interested in the effects of competition, and competition, in its purest form, is zero-sum. 
%
%
There is good reason, however, to view this zero-sum interaction not as between two \textit{players} but as between two \textit{teams}, with the members of each team being the \textit{genes} of the species (a point of view  influentially advocated by Dawkins~\cite{Dawkins1976}). The fact that the genes are, as ``agents", pursuing their optimization independently of one another, is crucial to the dynamics of sexual evolution.

\textbf{Dynamics:} The most tractable model of sexual evolution is to have the population at all times  be in a \textit{product distribution evolving according to the replicator equation}.
(The correspondence between this continuous dynamic and the discrete-time MWU was described in~\cite{Kleinberg09multiplicativeupdates,ITCS15MPP,PNAS2:Chastain16062014}; 
MWU is an ubiquitous meta-algorithm with numerous connections within the field of computer science~\cite{Arora05themultiplicative}.
In this paper we work directly in the replicator framework.)
The replicator equation~\cite{Taylor1978145,Schuster1983533}, given below in Eq.~\ref{eq:repli_A}, 
 is among the basic tools in mathematical ecology, genetics  
and the mathematical theory of  evolution.
Replicator dynamics and MWU have been studied extensively in numerous classes of games.
For example, in the case of potential games where the utilities of all agents are perfectly aligned, these learning processes are known to converge to Nash
equilibria and in fact generically to pure (non-randomized) Nash~\cite{Kleinberg09multiplicativeupdates}. On the contrary in zero-sum games \cite{Sato02042002,Akin84,Soda14} it is known that 
they can exhibit complex chaotic behavior, highly sensitive to initial conditions (butterfly effect). 

In our setting, an exceedingly rare form of structure manifests. 
\textit{We identify a novel class of conservative systems,}  analogous to Hamiltonian dynamics (\textit{e.g.} ideal pendulum).
  The conserved quantity does not have a meaning of energy. 
  This ``constant of the motion" is a new type of structure that has not been reported before that emerges from the combination of evolutionary dynamics and Boolean logic (Lemma \ref{lem:inv}). Possession of a conserved quantity is no reason to expect periodicity, as the dynamics are high dimensional.\footnote{$m+n$ variables, minus $1$ degree of freedom for the constant of the motion.}  The orbits are shown to be periodic via a novel type of embedding argument: Each orbit can be projected to 
 a, possibly different, planar (and thus periodic) conservative system without any loss of information.
 (Proposition \ref{prop:1} and Lemma \ref{lem:safe_planar}). 

\smallskip
{\bf Theorem 1:}
Given any two-team zero-sum game defined by two Boolean functions, so long as all  equilibria/fixed points are isolated, then all but a zero measure set of initial conditions lie on periodic trajectories of the replicator dynamics. (Formal statement in Section \ref{section:analysis}, Theorem \ref{thm:main1}).

  \smallskip
 
\begin{figure}[hbtp]
\centering
\begin{subfigure}{.5\textwidth}
  \centering
  \includegraphics[width=.7\linewidth]{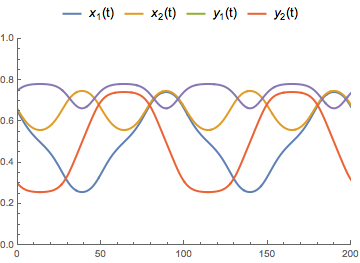}   
 \caption{XOR-XOR initial condition $(0.65, 0.66, 0.3, 0.75)$}
  \label{fig:sub1}
\end{subfigure}%
\begin{subfigure}{.5\textwidth}
  \centering
  \includegraphics[width=.7\linewidth]{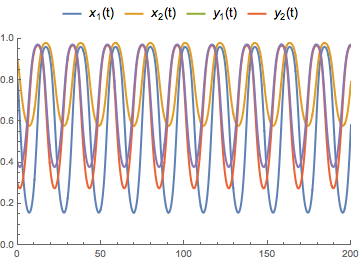}  
  \caption{OR-OR initial condition $(0.8, 0.9, 0.3, 0.4)$}
    \label{fig:sub2}
\end{subfigure}
\caption{Given generic initial conditions the system is periodic regardless of the pair of Boolean functions defining it.
The initial conditions, trajectories are of the form $(x_1(t),x_2(t),y_1(t),y_2(t))$. 
 }
\label{fig:test}
\end{figure}

 
 %
 %
An immediate corollary of our theorem is that since all monocultures are fixed points of the dynamics, 
almost all initial conditions lie on
trajectories that remain bounded away from the monocultures. 
%
It is interesting to note that although the above theorem is robust to the choice of the competing Boolean functions, \textit{e.g.}, XOR-XOR,
AND-AND, AND-XOR, etc., the  topology of the periodic orbits is highly sensitive to the choice of these functions (see Figure \ref{fig:test}).

 Besides proving periodicity, we analyze this phenomenon from a game-theoretic/optimality lens as well.
 We show that the time-average play over the strategy outcomes of the team zero-sum game is a correlated equilibrium. 
 Unlike zero-sum games, in team zero-sum games \cite{SchV17} (even Boolean team zero-sum games) Nash equilibria may include outcomes of widely varying utilities for each team (the notion of value no longer exists). 
Nevertheless, we show something extremely surprising about these evolutionary dynamics on 
 ``selfish" genes: sexual evolution leads to gene coordination that ``solves" the zero-sum competition at the species level (see Figure \ref{fig:test2}).
 Thus, in three different ways, sexual evolution implements an optimal strategy: 
 
 
 \smallskip
 
{\bf Theorem 2:}
Given any periodic orbit, the time average distribution of play over strategy outcomes converges point-wise to a {\bf correlated equilibrium of the team zero-sum game}.
Moreover, the time average Boolean output of each team converges to its {\bf unique Nash equilibrium strategy of the $2\times 2$ zero-sum game $U$}. The time average expected utility of each agent converges to the {\bf value of its team in the zero-sum game $U$}.  (Formal statement in Section \ref{section:time_average}, Theorems \ref{thm:correlated} and \ref{thm:main2}).

 \smallskip


\begin{figure}[hbtp]
\centering
\begin{subfigure}{.5\textwidth}
  \centering
  \includegraphics[width=.7\linewidth]{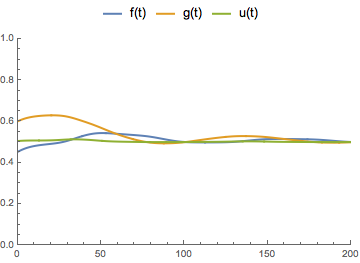}   
 \caption{XOR-XOR initial condition $(0.65, 0.66, 0.3, 0.75)$}
  \label{fig:sub1}
\end{subfigure}%
\begin{subfigure}{.5\textwidth}
  \centering
  \includegraphics[width=.7\linewidth]{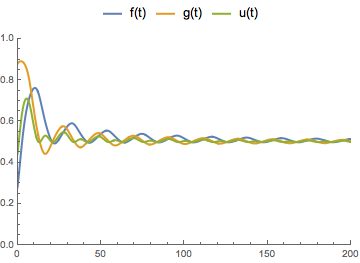}  
  \caption{OR-OR initial condition $(0.8, 0.9, 0.3, 0.4)$}
    \label{fig:sub2}
\end{subfigure}
\caption{The time averages of the outputs of both teams $(f,g)$ converge to the fully mixed NE of the game ($50\%, 50\%$). The time average of the expected utility of the first team (and hence the second as well) converges to the value of the game (which in this case is equal to $0.5$).
 }
\label{fig:test2}
\end{figure}

Sexual evolution solves the game optimally at the level of the species, despite the fact that it works on the level of the genes and that any signal coming from the species level  interaction is constantly getting scrambled due to genetic recombination.

 {\bf Structure of the paper:}  Section \ref{section:prelim} starts with an exposition of the model,  formal definitions of the game setting, as well as  the learning dynamic.
 Section \ref{section:analysis} contains the analysis of the periodicity of the model and starts with a high level description of the proof structure. 
 Section \ref{section:time_average} develops the connections to game theory, Nash and correlated equilibria and analyzes the time-average of payoffs and system behavior. The paper ends with a discussion section (Section \ref{section:Discussion}).
 The Appendix contains a section on background material on topology and dynamical systems.

\section{Preliminaries}
\label{section:prelim}

\paragraph{Notation:}  Vectors are in bold-face, and unless otherwise indicated are considered as column vectors. The transpose of $\xx$ is $\xx^T$. The $i^{th}$ coordinate of $\xx$ is $x_i$. 
$\xx_{-i}$ is the vector derived by removing $x_{i}$ from $\xx$.  
$|S|$ is the cardinality of set $S$. 
$\Delta(S)$ is the probability simplex with support set $S$, \textit{i.e.}, $\{(x_1,x_2,\dots x_{|S|}): x_i\geq 0~ \forall i,~ \sum_{i=1}^{|S|} x_i =1\}$.

\paragraph{
Zero-Sum Games among Species:}

We have two species, $A$ and $B$. 
$A$'s genome has $n$ genes, and $B$'s has $m$ genes. 
Both are haploids which means that each organism has one allele per gene. (This is the simplest possibility. Humans are diploid, and other numbers are possible.) The proportion of organisms of species $A$ that have allele $\gamma$ in their $i$-th gene
 is written $x_{i\gamma}$.
Similarly for $B$ the proportion of organisms  that have allele $\delta$ in their $j$-th gene is
 $y_{j\delta}$. Clearly, $\sum_{\gamma} x_{i\gamma}=1$ and $\sum_{\delta} y_{j\delta}=1$ for any $i$ or $j$.  
 In this paper we focus on the simplified setting that each gene of the host/parasite organism has two variants/alleles (\textit{i.e.}, allele $0$ and $1$). Thus, the genotype of an organism in $A$ (resp.\ $B$) is a Boolean vector of length $n$ (resp.\ $m$) 
and we will abbreviate notation by writing $x_{i0}$ as $x_{i}$.
The vector $\mbox{(\boldmath $x, y$)}=(x_1,\dots,x_n,y_1,\dots,y_m)$ expressing the composition of alleles in each population of organisms can be thought of as encoding a randomized strategy for each gene (the mixed strategy $(x_i,1-x_i)$ for gene $i$ in species $A$).

A more important simplifying assumption is that each genotype produces one of only two possible phenotypes 
(\textit{e.g.}, Rh blood factor $+/-$, long vs.\ short beak, etc.). A phenotype in our context is simply an arbitrary boolean function on the genome; for $s \in \{0,1\}^n$ we let $f(s)\in \{0,1\}$ be the phenotype of organism $A$, and likewise $g(\sigma)$ for $\sigma \in \{0,1\}^m$. The significance of this mapping is that organisms interact only through their phenotype. The payoff for this interaction is given by a 
 utility (or fitness) function $u$({\boldmath $s$},{\boldmath $\sigma$})=$U(f(${\boldmath $s$}),$g(${\boldmath $\sigma$})), where
$U$ is a payoff matrix of dimension $2\times2$.
 When organisms $s,\sigma$ interact, each gene in organism $A$ receives the same utility $u(${\boldmath $s$},{\boldmath $\sigma$}$)$ while each gene in organism $B$ receives the same utility  $-u(${\boldmath $s$},{\boldmath $\sigma$}$)$.
 
 \suppress{  The payoff matrix $U$ is a $2\times2$ matrix.
 %
$$
U =
 \begin{pmatrix}
  a & b \\
  c & d
 \end{pmatrix}
 $$
 }
 
 A natural example is as follows: $A$ is a parasite and $B$ is the host. 
 If the outcomes of functions $f, g$ match, \textit{i.e.}, the ``key" of the parasite matches the ``lock" of the host, then the utility of the host is $-1$ and the utility of the parasite is $+1$. Otherwise, the utilities are reversed.  In this case, matrix $U$ is the Matching Pennies payoff matrix.

$$u(\mbox{\boldmath $s,\sigma$}) = \begin{cases} 1, & \mbox{if } f(\mbox{\boldmath $s$}) = g(\mbox{\boldmath $\sigma$}) \\
                                                                        -1, & \mbox{otherwise.}  \end{cases}         
                                                                        \text{ , or equivalently }~~~                                                              
                                                                        U =
 \begin{pmatrix}
  1 & -1 \\
  -1 & 1
 \end{pmatrix}
$$

\noindent
More generally, we allow for the zero-sum game defined by $$
U =
 \begin{pmatrix}
  a & b \\
  c & d
 \end{pmatrix}
 $$ to be any 
$2\times2$ zero sum game that which has a unique Nash equilibrium, which is fully mixed.
(i.e., either
$\min(a,d)>\max(b,c)$ or $\max(a,d)<\min(b,c)$, or equivalently, the best response sequence cycles (clockwise/anti-clockwise) along the four outcomes.) \footnote{Otherwise, (given any generic $2\times2$ zero-sum game) this competition is trivial, since it can be solved via iterated elimination of strictly dominated strategies and convergence to equilibrium for the whole system follows from standard arguments. E.g., see~\cite{Weibull}. Our analysis applies for all rescaled $2\times2$ zero-sum that are not dominance solvable~\cite{Hofbauer98}.}

\paragraph{Replicator Dynamics and Weak Selection in the Evolution of Sexual Species}

If $p$ is the distribution on genomes $s$ of species $A$ and $q$ is the distribution on genomes $\sigma$ of species $B$, write $u_s=\sum_\sigma q(\sigma)u(s,\sigma)$ (the fitness of genome $s$) and $u=\sum_s p(s)u_s$. The replicator dynamics are that
 the rate of change of $p(s)$ is $\dot p(s)=p(s)(u_s-u)$. A few lines of calculation show that 
 the resulting rate of change of $x_{i}$ (the fraction of the population having allele $0$ in gene $i$) is
 \begin{equation}
\label{eq:repli_A}
\dot{x}_i = x_{i}(1-x_{i})\big(u_{i0} - u_{i1}\big) \quad \quad \text{ and } \dot x_i=0 \text{ if } x_i \in \{0,1\}
\end{equation} where 
 $u_{i0}=\frac{\sum_{s:s_i=0} p(s)u_s}{\sum_{s:s_i=0} p(s)}=\frac{\sum_{s:s_i=0} p(s)u_s}{x_i}$ being the fitness of allele $0$ of gene $i$.  Similarly for genes of species $B$,
 \begin{equation}
\label{eq:repli_B}
\dot{y}_j =-y_{j}(1-y_{j})\big(u_{j0} - u_{j1}\big)
\end{equation}

The work of Nagylaki in 1993  focused attention on study of these dynamics in the ``weak selection" model. In a sexual species, \textit{panmictic mating} (mating of individuals selected uniformly and independently) without selection pressures, leads over time to the genome distribution being a product distribution, that is, to $p(s)=\prod_i x_i^{1-s_i} (1-x_i)^s_i$. The weak selection model makes the approximation that selection is slow enough relative to reproduction that the genome may be considered at all times to be in a product distribution, with time-dependent marginals $x_i$ and $y_j$. The prior work~\cite{PNAS2:Chastain16062014,evolfocs14,ITCS15MPP,mehta_et_al:LIPIcs:2016:6407} is entirely within this model and it will be our focus as well. In weak selection the equations~\ref{eq:repli_A},~\ref{eq:repli_B} may be considered a complete description of the process rather than merely summary statistics.

What distinguishes our dynamics from prior work is that 
the fitness of a genotype is no longer a constant but depends on the composition of the population of the other species; the fitness of an allele of a particular gene depends on the compositions of both populations (except in that one gene).

\paragraph{Genes$\leftrightarrow$Agents, Species$\leftrightarrow$Team of Agents, Allele$\leftrightarrow$Strategy.}
 In terms of the analysis, it will be helpful  to think of the biological setting in purely game theoretic terms. The immediate game theoretic analogue of this setting is to study competitions between two teams, $A$ and $B$. The first team has $n$ agents, whereas the second has $m$. Each agent has two strategies, strategies $0$ and $1$ and we denote by $x_i\defeq x_{i0}$ the probability with which he chooses strategy $0$. Given a strategy outcome, all choices of agents in team $A$, (resp. $B$) are used as input in the Boolean functions of each team $f$ (resp. $g$) and each team participates with its respective output in the $2\times2$ zero-sum game $U$. All agents in a team enjoy exactly the same utility, \textit{i.e.}, their team's utility in game $U$. We denote by $u_i, u_{i0}, u_{i1}$ respectively the expected utility of agent $i$, the expected utility of agent $i$ given that he chooses $0$ and the expected utility of agent $i$ given that he chooses $1$ (where the randomness is over the product distribution over the mixed strategies of all agents).

\paragraph{Existence and Uniqueness of Global Solution}
The theory of differential equations ensures (for more details see, \textit{e.g.}, Chapter 6 of \cite{Weibull} or \cite{Arnold78}) that replicator dynamics of a multi-player game from initial conditions (i.e., initial probability distributions) $\zz_0$, have a unique global solution 
 $\Phi(\zz_0,\cdot):\Real\rightarrow \prod_i \Delta(S_i)$; furthermore that this solution is smooth as a function of  time and initial conditions.  
 We define a trajectory or orbit through an initial state $\zz_0$ as the image of the whole time axis under the solution mapping $\Phi(\zz_0,\cdot)$: 
 
 
\[\text{Traj}(\zz_0)=\{
\zz=\Phi(\zz_0,t) \text{ for some } t\in \Real\}\]

To ease notation, when keeping track of the initial condition is not critical, we  write {\boldmath $z$}$(t)$ instead of $\Phi(${\boldmath $z_0$}$,t)$. 
We  write {\boldmath $x$}$(t)$, (resp. $\yy(t)$) to denote the current (product) mixed strategy profile of genes in species $A$  (resp. $B$) or $x_i(t)$ (resp.\ $y_j(t)$) to denote the mixed strategy of a specific gene.

\section{Analysis}
\label{section:analysis}

\subsection{Overview of the Proof}

A critical insight is that instead of tracking the true state of the system {\boldmath $z$}$(t) =(${\boldmath $x$}$(t),${\boldmath $y$}$(t)) 
= (x_1(t), \dots, x_n(t), y_1(t), \dots, y_m(t))$, and directly trying to argue about the system in its native state space, we will focus on the quantities $\text{E}_{\sss \sim \xx} f(\sss)$ and 
$\text{E}_{\ssigma \sim \yy} g(\ssigma)$, \textit{i.e.}, the expected output of the Boolean functions of both teams of agents. We will denote these quantities as $f,g$.\footnote{The expected output of the Boolean function of team $A$, $f(${\boldmath $x$}$(t))$, is clearly a function of {\boldmath $x$}$(t)$, however, sometimes to simplify notation we will just write $f(t)$ when we wish to focus on the time dependency, or just $f$.} 
As it turns out, it will be convenient to think of the distribution $(f,1-f)$ as encoding a mixed strategy implemented in the $2\times2$ zero-sum game $U$. 




In Section \ref{subsection:eq}, we identify and classify the equilibria (fixed points) of the dynamics. These equilibria can be grouped into two  categories. {\bf Nash fixed points} are states in which the expected output of each Boolean function $(f, g)$ encodes the unique (fully mixed) Nash equilibrium of the $2\times2$ zero-sum game $U$.  These are stationary due to the fact that no team can deviate and gain the upper hand on its opposing team. The second type of fixed points, are states where at least one of  the two teams got ``stuck". Their opposing team (\textit{e.g.}, team $B$) is not (necessarily) implementing its minimax strategy but nevertheless no agent of team $A$ can influence  the expected outcome of his team's Boolean function via unilateral deviations. One such example is when team $A$ is implementing a XOR function and at least two agents choose between $0,1$ uniformly at random. We call these fixed points {\bf strange fixed points} as they intuitively correspond to evolutionary flukes, where changes to no single gene can affect the composition of the species. 

In Section \ref{subsection:Independent}, we focus on a single population and show that the corresponding vector field can be expressed as a product of a scalar ``rate" term (that depends on the mixed strategy of the opposing species) and a vector that only depends on the team's own behavior (Proposition \ref{prop:AindB}). Thus, the trajectory that each team traverses depends only the Boolean function that it implements (e.g., $f$) and its own initial condition (e.g., $x_1, x_2,\dots, x_n$). Effectively, this trajectory corresponds to replicator dynamics in a common utility/potential game where the joint utility of each agent in any mixed strategy outcome is the expected output  $f$ of the team's Boolean function. 

This connection to potential games becomes handy for several reasons. First, in Theorem \ref{thm:safe}, we prove that strange fixed points are indeed evolutionary flukes that can be ruled out under typical genericity conditions. Secondly, 
in Proposition \ref{lem:chasing}, we can already establish that $f, g$ 
exhibit a natural ``chasing" relationship, which is an early step towards proving periodicity. That is, if we interpret the expected output of each organism as a mixed strategy with which the organism participates in
the zero-sum game then this mixed strategy will move in the direction that would have increased its expected payoff given the mixed strategy of the opposing species. 
This creates a ``chasing" behavior with the directionality of the movement of each mixed strategy (\textit{i.e.}, increasing or decreasing) flipping when the opposing team's strategy transitions through the unique (mixed) Nash equilibrium of the game.


Furthermore, in Section \ref{subsection:Reduction}, we leverage the connections to potential games to formally prove that it suffices to keep track only of the quantities $f,g$.
 As explained above, these quantities correspond to the potential function
 in a common utility/potential game where the joint utility of each agent in any mixed strategy outcome is the expected output  of each team's Boolean function. 
In such a potential game, along any nontrivial trajectory the common utility/potential is strictly increasing with time and thus given any initial condition $\xx_0$
 there exists a bijective function between the time range over which the trajectory is defined, $(-\infty, \infty)$, and the range of potential values over this trajectory. Thus, given the initial condition of the team's behavior
  and the current output of the team, $\textit{e.g.},$ $f(t)$, the current behavior of each member of the team $x_i(t)$ is uniquely defined. In a sense, each trajectory can be embedded onto a two dimensional system since given the initial conditions of both teams as long as we keep track of each team's expected output $f,g$ we can 
  uniquely identify the exact state of the system. 

In order to prove the system periodicity we need to establish connections to the theory of Hamiltonian systems. In Section \ref{subsection:Hamiltonian}  we  show that the two dimensional system that couples $f, g$ together is effectively a conservative system that preserves an energy-like function. Specifically, up to team specific reparametrizations and change of variables, the dynamical system has the form

\[ 
\left \{
  \begin{tabular}{ccc}
  $\frac{df}{dt}$& = & $g$ \\
$\frac{dg}{dt}$&=&$-f$ 
  \end{tabular}
\right \}
\Leftrightarrow
\left \{
  \begin{tabular}{ccc}
  $\frac{df}{dt}$& = & $\frac{\partial H}{\partial g}$ \\
$\frac{dg}{dt}$&=&$-\frac{\partial H}{\partial f}$ 
  \end{tabular}
\right \}
\]

%
%
\noindent
which is a standard Hamiltonian system with Hamiltonian function equal to $H=\frac{f^2+g^2}{2}$. In this parameterization all trajectories are cycles centered at $0$.
In our case,  this conserved quantity, or ``constant of the motion", is  
more elaborate; 
nevertheless, once we establish it exists, we can leverage  standard tools from topology of dynamical systems (Appendix \ref{Section:Background}), and establish that its orbits are periodic (Theorem \ref{thm:periodic}).
Putting everything in this section together, Theorem 1 (more formally given as Theorem \ref{thm:main1}) follows. 

Finally, in Section \ref{section:time_average} we investigate the game-theoretic properties of these periodic orbits. By periodicity, the time-averages of all involved quantities are well defined. In Theorem \ref{thm:correlated} we show that the time-average play over the strategy outcomes of the team zero-sum game is a correlated equilibrium of that game. Unlike zero-sum games, in team zero-sum games \cite{SchV17} (even Boolean team zero-sum games)  their sets of (Nash) equilibria may include outcomes of widely varying utilities for each team. Nevertheless, in Theorem 2 (given more formally as \ref{thm:main2}), we establish that sexual evolution leads to gene coordination at the species level, in a time-averaged sense. Namely, time average of the output of each team $(f,g)$
is equal to the unique fully mixed Nash equilibrium strategy of that team in the $2\times2$ zero-sum game $U$. Furthermore,
 the average utility of each team (and hence all of its members)  is exactly equal to its value in the two species zero-sum competition. 

\subsection{System Description \& Fixed Points: The Nash, the Strange \& the Partial}  
\label{subsection:eq}

\begin{lem} 
\label{lem:sys}
There exists a constant $\alpha \neq 0$ such that the replicator system equations reduce to: 
$$\dot{x}_i=\alpha x_i(1-x_i)(f_{i0}-f_{i1})(g-q) \text{ for all agents $i$ in team } A,$$ 
$$\dot{y}_j=-\alpha y_j(1-y_j)(g_{j0}-g_{j1})(f-p) \text{ for all agents $j$ in team } B$$ 

\noindent
where $(p, 1-p)$, $(q, 1-q)$ the unique fully mixed Nash equilibrium strategy in game $U,-U^T$
and where $f_{i\gamma}=\text{E}_{\sss_{-i}\sim \xx_{-i}} f(\gamma,\sss_{-i})$, 
$g_{j\gamma}=\text{E}_{\ssigma_{-j}\sim \yy_{-j}} g(\gamma,\ssigma_{-j})$ for $\gamma \in \{0,1\}$.
\end{lem}

The proof of Lemma \ref{lem:sys} deferred to Appendix \ref{Appendix:ProofSys}.
 This reparametrization does not affect the shape of the trajectories. 
Hence, we assume wlog  $\alpha$ to be equal to $1$.
\medskip

\noindent
{\bf Structure of fixed points.} 
 Amongst all equilibria with full support, there exist two different types of fixed points.
The first type corresponds to outcomes where $f=p$ and $g=q$, \textit{i.e.}, outcomes in which the expected output of each Boolean function encodes the unique (fully mixed) Nash equilibrium of the $2\times2$ zero-sum game. We call these {\bf Nash fixed points}. The second type of fixed points, we have either  $f_{i1}=f_{i0}$ for all agents $i$ of team $A$, or  $g_{j1}=g_{j0}$ for all agents $j$ of team $B$, or both. We call these fixed points, {\bf strange fixed points}.\footnote{It may be the case that one fixed point satisfies the definition of both Nash fixed point as well as strange fixed point. E.g. when both teams $A,B$ implement the XOR function and play against each other in a Matching Pennies game then the uniformly mixed strategy is both a Nash fixed point and a strange fixed point. This non-exclusivity simplifies the exposition and thus we allow it.}
These are fixed points where at least one of  the two teams got ``stuck". Their opposing team (\textit{e.g.}, team $B$) is not (necessarily) implementing its minimax strategy but nevertheless no agent of team $A$ can influence  the expected outcome of his team's Boolean function via unilateral deviations. One such example is when team $A$ is implementing a XOR function and at least two agents choose between $0,1$ uniformly at random. 

Finally there exist fixed points in which some agents  are using pure strategies (\textit{e.g.}, $x_i,y_j =0$ or $1$).  
We call these {\bf partial support fixed points}.
We can complete the categorization by defining as \textit{partial support Nash fixed points}, (resp. \textit{partial support strange fixed points})  those partial support fixed points with at least one randomizing agent
  such that when examining the subgame defined by those strategies played with positive probability then they encode  a Nash (resp. strange) fixed point. Fixed points without any randomizing agents are called pure fixed points.

\subsection{The Topology of the Trajectory of Team $A$ is Independent of Team $B$}  
\label{subsection:Independent}

\begin{prop}
\label{prop:AindB}
The trajectory of team $A$ $\{\xx \in  [0,1]^n:
(${\boldmath $x,y$}$)=\Phi((\xx_0,\yy_0),t) \text{ for some } t\in \Real\}$
%
is a subset of the trajectory of system
  $$\dot{x}_i= x_i(1-x_i)(f_{i0}-f_{i1})\text{ for all agents in team } A,$$ 
  
  \noindent
  with initial condition $\xx_0$, which is independent of team $B$. 
  We call this system, team's $A$ subsystem and we denote its solution by $\Phi^A(${\boldmath $x$}$,t)$.
  Moreover, $f$ is a {\bf strict Lyapunov function} in this system. That is, given any initial condition $\mbox{\boldmath $x=x_0$}$ we have  
$\left.
   \frac{df}{dt}
 \right|_{\xx=\xx_0}>0,$
 unless $\xx_0$  is a fixed point of $\Phi^A$.
\end{prop}

\begin{proof}
The multiplicative term $ (g-q)$ is common across all terms of the vector field corresponding to agents in team $A$ in lemma \ref{lem:sys}. Hence, it dictates the magnitude of the vector field (the speed of the motion), but does not affect directionality other than moving backwards or forwards along the same trajectory. 
Specifically, both systems 
$\dot{x}_i= x_i(1-x_i)(f_{i0}-f_{i1})\text{ for all agents in team } A,$ and $\dot{x}_i= -x_i(1-x_i)(f_{i0}-f_{i1})\text{ for all agents in team } A$ have exactly the same orbits (but traverse them in opposite direction). So, 
 the trajectory of team $A$ in our original system corresponds to a subset of a specific orbit of subsystem $A$. 
 This specific orbit  starts at the initial condition  $\xx_0$ ($A$-team's  initial condition in the original system).

Moreover, we will show that $f$ is a strict Lyapunov function for this projected system.
Due to the multilinearity of $f$ on $x_i$'s: $f= x_if_{i0}+(1-x_i)f_{i1}$ and therefore  $\frac{\partial f}{\partial x_i}= f_{i0}-f_{i1}$.
Combining this with the definition of the vector field of subsystem $A$ we have that

$$\frac{df}{dt} =\sum_i \frac{\partial f}{\partial x_i}\dot{x_i}= \sum_i  x_i (1-x_i)\big(f_{i0}-f_{i1}\big)^2 $$

\noindent
The summation is clearly nonnegative and it is only equal to zero at the fixed points of $\Phi^A$.
\end{proof}

\noindent
{\bf Connection to potential games.} Team's $A$ subsystem is equivalent to applying replicator dynamics to a partnership game (a game where all $i \in \{1,2,\dots, n\} $ agents receive the same payoff/utility at each (mixed) outcome) where the common utility function at mixed strategy  {\boldmath $x$} is equal to $f(${\boldmath $x$}$)=\text{E}_{\sss \sim \xx}f(s)$.  
A partnership game is a potential game with potential function equal to the common utility. The potential is a strictly increasing function along any non-trivial system trajectory and all initial conditions implying convergence to equilibria (see \textit{e.g.}, \cite{Kleinberg09multiplicativeupdates}). We will leverage this connection and our current understanding of replicator dynamics in potential games from  \cite{Kleinberg09multiplicativeupdates} to argue that in system $\Phi$ only a measure zero set of initial conditions may converge to strange fixed points.

\begin{definition}
 We call an initial condition $\xx_0$ of subsystem $\Phi^A$ {\bf safe}, if and only if, the orbit $\Phi^A(\xx_0,\cdot)$ does not converge to a non-pure\footnote{non-pure = a fixed point with at least one randomizing agent.}  
 fixed point for $t \rightarrow \pm \infty$. Analogous definitions apply for subsystem $\Phi^B$ as well.
 We call an initial condition $(\xx_0,\yy_0)$ of system $\Phi$ {\bf safe}, if and only if, both $\xx_0$ and $\yy_0$ are safe in their respective subsystems.
\end{definition}

By proposition \ref{prop:AindB}, given a safe initial condition $(\xx_0,\yy_0)$, its respective orbit $\Phi((\xx_0,\yy_0),\cdot)$ clearly cannot converge to a partial strange fixed point for $t \rightarrow \pm \infty$. Indeed,
 a necessary condition for convergence to a partial strange fixed point in $\Phi$ given initial condition $(\xx_0,\yy_0)$ as  $t \rightarrow \pm \infty$ is that either $\Phi^A(\xx_0,\cdot)$ or  $\Phi^B(\yy_0,,\cdot)$ converge to a non-pure  
  fixed point as $t \rightarrow \pm \infty$. 

\begin{thm}
\label{thm:safe}
If the fixed points of $\Phi$ are isolated then all but a measure $0$ set of its initial conditions  are safe.
\end{thm}

The proof of theorem \ref{thm:safe} is deferred to appendix \ref{Appendix:ProofSafe}.
At this point we can show that the expected outputs $f,g$ of the two teams when viewed as mixed strategies (\textit{i.e.}, probability distributions) $(f,1-f)$, $(g,1-g)$ in the $2\times2$ zero-sum game are being updated in a ``rational" way. Specifically, when they are updated according to the system equations they will develop a ``chasing" behavior where each mixed strategy will move towards the direction that myopically increases its expected payoff in the zero-sum game. This statement in itself does not suffice to argue periodicity  as one can easily create trajectories (not of our dynamics of course) that
 spiral towards the fully mixed Nash equilibrium or diverge to the boundary while displaying this chasing behavior.

\begin{prop}
\label{lem:chasing}
Unless no agent of team $A$ can influence the output of $f$ via unilateral deviations, the expected output of team's $A$ Boolean function, $f$, will increase (decrease) if and only if the output $g$ of team $B$ is larger (smaller) than $q$ ($q=$ the probability of choosing the first action in its unique Nash equilibrium of the zero-sum game). Similarly, $g$ will decrease (increase) when the output $f$ of team $A$ is larger (smaller) than $p$ ($p=$ the probability of choosing the first action in its unique Nash equilibrium of the zero-sum game).
\end{prop}

\begin{proof}

Due to the multilinearity of $f$ on $x_i$'s: $f= x_if_{i1}+(1-x_i)f_{i0}$ and therefore  $\frac{\partial f}{\partial x_i}= f_{i1}-f_{i0}$.
Combining this with the vector field form presented in lemma \ref{lem:sys} we have that

$$\frac{df}{dt} =\sum_i \frac{\partial f}{\partial x_i}\dot{x_i}=  (g-q)\sum_i x_i (1-x_i)\big(f_{i1}-f_{i0}\big)^2$$

\noindent
The summation is clearly nonnegative.
 In fact, it is only equal to zero at the fixed points of team's $A$ subsystem. 
\end{proof}

From this point forward we will  focus on {safe} initial conditions. This is a full measure set within the set of all initial conditions.
We will prove that any such  state is periodic, \textit{i.e.}, it lies on a closed orbit by establishing connections to planar Hamiltonian systems.  

\subsection{Reduction to $2$-Dimensional Systems via Competing Lyapunov Functions}  
\label{subsection:Reduction}

The next proposition states that knowledge of the initial  condition as well as of the evolving values of $f,g$ suffices (in principle) to recover the complete system state at any time $t$.

\begin{prop}
\label{prop:1}
Given a safe initial condition $\zz_0=(\xx_0,\yy_0)$, of system $\Phi$ as well that the values  $f(t),g(t)$, 
 there exist smooth functions $X^{\xx_0}_i,Y^{\yy_0}_j: [0,1] \rightarrow [0,1]$ such that $x_i(t)= X^{\xx_0}_i(f(t))$ $($resp.  $y_j(t)= Y^{\yy_0}_j(g(t))$$)$ for all $t\in \Real$
 and for each agent $i$ of team $A$ $($resp. for each agent $j$ of team $B)$.
\end{prop}

\begin{proof}
If $\xx(t)$ (resp. $\yy(t)$) is time invariant,  then the problem is trivial. Suppose not.
We will argue that given an initial condition  for every agent $i$ of team $A$, \textit{i.e.}, $\xx_0$, its  mixed strategy at time $t$, as captured by $x_i(t)$ is uniquely defined given $f(t)$. We know that the curve traced by the agents of the first team {\boldmath $x$}$(t)$  is defined by their initial conditions and the function that they implement. Specifically, it is  contained in the trajectory of team's $A$ subsystem.
So, as long as we can uniquely pinpoint a state $\xx$ on subsystem's $A$ trajectory $\Phi^A(\xx_0,\cdot)$,  given $\xx_0$
 and $f(t)$, then  this must correspond to  team's  $A$ state in the complete system $\Phi((\xx_0,\yy_0),\cdot)$.
 Moreover, in subsystem $A$, $\frac{df}{dt}>0$ unless we are at a fixed point. However, since $\xx(t)$ is not time-invariant, $\xx_0$ is not a fixed point of subsystem $A$.
Finally, by the uniqueness of the system solutions, we cannot reach a fixed point in finite time, and hence $\frac{df}{dt}>0$ for all times $t\in \Real.$ So $f$ as a function of time in subsystem $A$ is always increasing, and it is smooth since $f(\Phi^A(\xx_0,\cdot))$ is a composition of smooth functions. Thus, by the inverse function theorem (see Appendix \ref{Section:Background}) $f^{-1}_{\xx_0}$ exists\footnote{The inverse function of $f$ is effectively parametrized by $\xx_0$ and that's why we write $f^{-1}_{\xx_0}$.}, is smooth,  and  
 is strictly increasing. Thus, it is bijective between its domain, $\Big(\lim_{t\rightarrow -\infty} f(\Phi^A(\xx_0,t)), \lim_{t\rightarrow +\infty} f(\Phi^A(\xx_0,t)) \Big)=\footnote{Since $(\xx_0,\yy_0)$ is safe, in subsystem $A$ the only possible asymptotes for $f$ are $0$ and $1$.}(0,1)$, and $\Real$ and given an input $v=f(\Phi^A(\xx_0,t))$ in its domain
$f^{-1}_{\xx_0}$ returns $t$, the unique time instance at which the Lyapunov function in subsystem $A$ attains value $v$ given initial condition $\xx_0$. Thus, $f^{-1}_{\xx_0}$  is well defined given $\xx_0$ alone. 
We define as $Proj_i$ the projection function that given a vector returns its $i$-th element, \textit{i.e.}, $Proj_i(\xx)=x_i$.
Putting everything together, $X^{\xx_0}_i\defeq Proj_i\Big(\Phi^A\Big(\xx_0,f^{-1}_{\xx_0}(\cdot)\Big)\Big)$ indeeds recovers the accurate state of agent $i$ in team $A$ given the current value of $f$ in system $\Phi$, \textit{i.e.}, $f\big(\Phi\big((\xx_0,\yy_0),t\big)\big)$, since it lies in $\Big(\lim_{t\rightarrow -\infty} f(\Phi^A(\xx_0,t)), \lim_{t\rightarrow +\infty} f(\Phi^A(\xx_0,t)) \Big)=(0,1)$. Finally, we continuously extend  $X^{\xx_0}_i$ to $[0,1]$.
\end{proof}


We return to the study of the two team system 
and argue about the periodicity of its orbits $(f(t),g(t))$. Due to the existence of functions $X^{\xx_0}_i,Y^{\yy_0}_j$, mapping  $(f(t),g(t))$ to a unique $(\mbox{\boldmath $x$}(t),\mbox{\boldmath $y$}(t))$ the periodicity of $(f(t),g(t))$ extends to the system trajectories of $\Phi$.
To simplify notation will we write from now on $X_i,Y_j$ instead of $X^{\xx_0}_i,Y^{\yy_0}_j$ but the dependency on $\xx_0$, $\yy_0$ should be kept in mind. 

\begin{definition}{\bf $(p,q,r,w)$-planar dynamical system:} We define the following class of planar dynamical systems on $[0,1]^2$ parametrized by a point $(p,q) \in (0,1)^2$ and two smooth functions $r,w$ defined on $[0,1]$ with  $r(0)=r(1)=w(0)=w(1)=0$ that are 
 strictly positive in $(0,1)$. 
Given such $p,q,r,w$ we define a $(p,q,r,w)$-planar dynamical system as follows:
   \begin{eqnarray*}
   \frac{d\xi}{dt}&=&r(\xi)(\zeta - q)\\
   \frac{d\zeta}{dt}&=&-w(\zeta)(\xi - p)
   \end{eqnarray*}
\end{definition}


The existence, uniqueness and smoothness of global solutions 
in the case of  $(p,q,r,w)$-planar dynamical systems
follows from standard arguments, since the compact region $[0,1]^2$ is invariant 
and the vector field is smooth. (see \textit{e.g.}, $\cite{Arnold78}$)

\begin{lem}
\label{lem:safe_planar}
 Given a safe  initial condition,  $(\xx_0,\yy_0)$ of $\Phi$ there exists a 
$(p,q,r,w)$-planar dynamical system
such that if $(\xi_0,\zeta_0)=\big(f(\xx_0),g(\yy_0)\big)$ then $(\xi(t),\zeta(t))=\big(f($\mbox{\boldmath $x$}$(t))$,
$g($\mbox{\boldmath $y$}$(t))\big)$ for all $t\in \Real$.  
\end{lem}

\begin{proof}
For our two team systems we have  that 

\begin{eqnarray*}
\lefteqn{\frac{df}{dt} =\sum_i \frac{\partial f}{\partial x_i}\dot{x_i}= (g-q)\sum_i x_i (1-x_i)\big(f_{i0}-f_{i1}\big)^2}\\
                                      &=&  (g-q) 
                                       \underbrace{
                                      \sum_i X_i\big(f(t)\big)
                                       \Big(1-X_i\big(f(t)\big)\Big)\Big(
                                      \text{E}_{\sss_{-i}\sim X_{-i}(f(t))} f(0,\sss_{-i})- \text{E}_{\sss_{-i}\sim X_{-i}(f(t))} f(1,\sss_{-i})
                                      \Big)^2}_\textit{r(f)}
\end{eqnarray*}

\noindent
where   $r(f)=\sum_i X_i\big(f(t)\big) \Big(1-X_i\big(f(t)\big)\Big)\Big(
                                      \text{E}_{\sss_{-i}\sim X_{-i}(f(t))} f(0,\sss_{-i})- \text{E}_{\sss_{-i}\sim X_{-i}(f(t))} f(1,\sss_{-i})
                                      \Big)^2$ is a smooth  function and is clearly nonnegative since $X_i\big(f(t)\big)\in [0,1]$. 
                                      If $f(t)=0$, \textit{i.e.}, $\text{E}_{{\sss\sim \xx}} f(s)=0$, since $\xx$ corresponds to a product distribution
                                      and $f$ is a Boolean function we have that 
                                      for each $i$ either $x_i=0$ or $x_i=1$, or  the value of $f$ over all outcomes in the support 
                                      of $x$ are equal to $0$. Hence,  $\text{E}_{\sss_{-i}\sim X_{-i}(f(t))} f(1,\sss_{-i})=\text{E}_{\sss_{-i}\sim X_{-i}(f(t))} f(0,s_{-i})=0$.
                                      Thus, in all cases if $f(t)=0$ then $r(f)=\sum_i X_i\big(f(t)\big) \Big(1-X_i\big(f(t)\big)\Big)\Big(
                                      \text{E}_{\sss_{-i}\sim X_{-i}(f(t))} f(0,\sss_{-i})- \text{E}_{\sss_{-i}\sim X_{-i}(f(t))} f(1,\sss_{-i})
                                      \Big)^2=0$. A similar argument can be applied if $f(t)=1$. 
                                      Finally, we will argue that since $\xx_0$ is safe then $r(f)$ is strictly positive for $f\in (0,1)$. 
                                      It suffices to show that 
                                      for any $f(t)\in (0,1)$,  $\sum_i X_i\big(f(t)\big) \Big(1-X_i\big(f(t)\big)\Big)\Big(
                                      \text{E}_{\sss_{-i}\sim X_{-i}(f(t))} f(0,\sss_{-i})- \text{E}_{\sss_{-i}\sim X_{-i}(f(t))} f(1,\sss_{-i})
                                      \Big)^2 >0$. Since $f(t)\in (0,1)$ there must exist some randomizing agents in distribution {\boldmath $x$}. Amongst these agents, there must exist an agent $i$, 
                                      $
                                      \text{E}_{\sss_{-i}\sim X_{-i}(f(t))} f(0,\sss_{-i})- \text{E}_{\sss_{-i}\sim X_{-i}(f(t))} f(1,\sss_{-i})
                                      \neq 0$, since otherwise the state where each agent $i$ of team $A$ plays $X_{i}(f(t))$ would be a non-pure 
                                       fixed point of subsystem $A$ contradicting the assumption that the initial condition $(\xx_0,\yy_0)$ is safe.        
                                      The argument for $g$, \text{i.e.}, team $B$ follows along the same lines as for team $A$.                               
                                      \end{proof}
%
%
\subsection{Constants of Motion, Hamiltonian Systems, and Periodicity}  
\label{subsection:Hamiltonian}

Finally, we establish that the two dimensional system that couples $f, g$ together is effectively a conservative system that preserves an energy-like function. It is easy to check that if in the definition of our $(p,q,r,w)$-planar dynamical system we set the functions $r,w$ being everywhere equal to $1$ then after a change of variables $f=\zeta-q, g=\xi-p$, the dynamical system has the form: $ \frac{df}{dt}=g= \frac{\partial H}{\partial g}, \frac{dg}{dt}=-f= -\frac{\partial H}{\partial f}$,
which is a prototypical Hamiltonian system with Hamiltonian function equal to $H=\frac{f^2+g^2}{2}$. All of its trajectories are cycles centered at $0$.

The conserved quantity, i.e., the ``constant of the motion", in our case is described in the following lemma and
leveraging it we will establish that the system trajectories are periodic.

\begin{lem}
\label{lem:inv}
The quantity $$H(\xi,\zeta)=\int^\xi_p \frac{z-p}{r(z)}dz + \int^\zeta_q \frac{z-q}{w(z)}dz$$ is a
constant of the motion (first integral) of the $(p,q,r,w)$-planar dynamical system, \text{i.e.},
it is time-invariant given any initial condition.
\end{lem}

\begin{proof} By applying the chain rule on $\frac{dH(\xi,\zeta)}{dt}$ we have:\\
$\frac{d\big(\int^\xi_p \frac{z-p}{r(z)}dz\big)}{dt} + \frac{d\big(\int^\zeta_q \frac{z-q}{w(z)}dz\big)}{dt}= \frac{\xi-p}{r(\xi)}\frac{\partial\xi}{\partial t}+\frac{\zeta-q}{w(\zeta)}\frac{d\zeta}{\partial t}=(\xi-p)(\zeta-q)-(\zeta-q)(\xi-p)=0.$
\end{proof}

\begin{thm}
\label{thm:periodic}
  Every safe initial condition  $(\xx_0,\yy_0)$ lies on a periodic orbit of $\Phi$.
\end{thm}

\begin{proof}
From lemma \ref{lem:safe_planar} the system corresponds to a $(p,q,r,w)$- planar system.
If  $(\xx_0,\yy_0)$  is a Nash fixed point of $\Phi$ then 
it is trivially a periodic point.
Suppose $(\xx_0,\yy_0)$ is not a Nash fixed point, then either $f\neq p$ or $g \neq q$ (or both). 
In all cases  $H\big(f(\xx_0),g(\yy_0)\big)>0$ and due to lemma \ref{lem:inv} the trajectory of the planar system
stays bounded away from its unique interior equilibrium $(p,q)$, since $H(p,q)=0$.
 Moreover, the gradient of $H$ at $(\xi,\zeta)$ is equal to $(\frac{\xi-p}{r(\xi)},\frac{\zeta-q}{w(\zeta)})$
 and thus we  can create a trapping/invariant region $C=\{(x,y): 0<\alpha<H(\xx,\yy)<\beta\}\subset (0,1)^2\setminus (p,q)$. 
By the Poincar\'{e}-Bendixson theorem and since the trapping (invariant) regions does not contain any fixed points the $\alpha,\omega$-limit set of the trajectory is a periodic orbit.
Since the gradient of $H$ is only equal to $\bf 0$ at $(p,q)$, $H\big(f(\xx_0),g(\yy_0)\big)$ is a regular value of $H$. 
By the regular value theorem $H^{-1}\Big(H\big(f(\xx_0),g(\yy_0)\big)\Big)$ is a manifold of dimension $1$.
 The union of the trajectory starting at $(\xx_0,\yy_0)$, along with its $\alpha,\omega-$limit sets, is a closed, connected 1-manifold and thus it is isomorphic to 
$S^1$ (see Appendix \ref{Section:Background}). 
\end{proof}



At this point, we are ready to piece together all our structural characterizations of the system trajectories to derive our first main theorem:

\begin{thm}
\label{thm:main1}
Given any two-team zero-sum game defined by two Boolean functions, so long as all  equilibria/fixed points are isolated, then all but a zero measure set of initial conditions lie on periodic trajectories of the replicator dynamics
\end{thm}

\begin{proof}
By theorem \ref{thm:safe} since all the system equilibria are isolated, then all but a measure zero set of initial conditions are safe. By lemma \ref{lem:safe_planar} and \ref{thm:periodic}, the projection of these trajectories on the space of outputs of each of the two teams, \textit{i.e.} on the space $(f, g)$ is periodic. Finally, by \ref{prop:1} given a safe initial condition $\zz_0=(\xx_0,\yy_0)$, of system $\Phi$ as well that the values  $f(t),g(t)$, 
 there exist smooth functions $X^{\xx_0}_i,Y^{\yy_0}_j: [0,1] \rightarrow [0,1]$ such that $x_i(t)= X^{\xx_0}_i(f(t))$ $($resp.  $y_j(t)= Y^{\yy_0}_j(g(t))$$)$ for all $t\in \Real$
 and for each agent $i$ of team $A$ $($resp. for each agent $j$ of team $B)$, thus the periodicity of $f(t), g(t)$ translates to a periodic orbit on the space of system behaviors, \textit{i.e.} on $(\xx(t), \yy(t))$ and the proof is complete.
\end{proof}

\section{Time Averages, Connections to Equilibria and Utility}
\label{section:time_average}

Next, we will show that the time-average of the periodic trajectories of replicator dynamics satisfies some interesting game theoretic properties. In order to discuss these properties, it is useful to provide a reminder on some of the most basic solution concepts in game theory.

We give the definition of a correlated equilibrium, from \cite{aumann1974subjectivity}.

\begin{definition} A \emph{correlated equilibrium} (CE) is a distribution \( \pi \) over the set of action profiles \( S = \prod_{i} S_i \) such that for all player \( i \) and strategies \( s_i, s_i' \in S_i, s_i \neq s_i' \),
	\[
		\sum_{s_{-i} \in S_{-i}} u_i(s_i, s_{-i}) \pi(s_i, s_{-i})
		\geq \sum_{s_{-i} \in S_{-i}} u_i(s_i', s_{-i}) \pi(s_i, s_{-i})
	\]
\end{definition}

\noindent
We will also make use of  the coarse correlated equilibrium (\cite{young2004strategic}), which is exactly the set of distribution that no-regret algorithms converge to. This convergence is only set-wise, \textit{i.e.}, distance of the time average behavior of no-regret dynamics and the set of CCE converges to zero, however, the time-average play may never converge to a specific CCE.

\begin{definition}
	A \emph{coarse correlated equilibrium} (CCE) is a distribution \( \pi \) over the set of action profiles \( S = \prod_{i} S_i \) such that for all player \( i \) and strategy \( s_i \in S_i \),
	\[
		\sum_{s \in S} u_i(s) \pi(s)
		\geq \sum_{s_{-i} \in S_{-i}} u_i(s_i, s_{-i}) \pi_i(s_{-i})
	\]
	where \( \pi_i(s_{-i}) = \sum_{s_i \in S_i} \pi(s_i, s_{-i}) \) is the marginal distribution of \( \pi \) with respect to \( i \).
\end{definition}

First, we will show that the time-average distribution over the space of strategy outcomes over any periodic orbit is a coarse correlated equilibrium. Furthermore, the time-average of the output of each team $f, g$ corresponds to the unique Nash equilibrium of the $2\times2$ zero-sum game. Finally, the expected utilities of all agents correspond to the value of each of their respective teams in their zero-sum game. In effect, the sexual replicator dynamics enable the agents of each team to collaborate with each other so as to optimally solve the zero-sum game against the opposing team. Figure \ref{fig:test2} shows specific examples of periodic trajectories where time-averaging over them converges to the solution of a Matching Pennies game\footnote{This is a rescaled Matching Pennies game where all utilities are nonnegative and the value of the game is $0.5$.} between two teams. 

\begin{thm}
\label{thm:correlated}
Given any periodic orbit, the time average distribution of play over strategy outcomes converges point-wise to a specific correlated equilibrium.
\end{thm}

\begin{proof}
The time average of play is well defined and converges to a unique distribution over the space of strategy outcomes. Since the trajectory is periodic and the interior of the state space is invariant, the trajectory stays bounded away from the boundary of the state space. However, in this case  the time average of the trajectory of the replicator converges to a coarse correlated equilibrium. More specifically, we will show that if any individual agent deviates to any fixed strategy then the time average of his expected utility does not decrease.\footnote{Interestingly,  we will show that it does not decrease either. The time average any agent's expected utility remains invariant given any deviation to a fixed strategy.}  The replicator equation 
$\dot{x}_i = x_{i}(1-x_{i})\big(u_{i0} - u_{i1}\big)$  is equivalent to $\dot{x}_i = x_{i}\big(u_{i0} - \hat{u}_{i}\big)$ as well as 
$\dot{x}_i = -(1-x_{i})\big(u_{i1} - \hat{u}_{i}\big)$, where
$\hat{u}_{i} = x_{i}u_{i0} + (1-x_{i})u_{i1}$, \textit{i.e.}, the expected utility of agent $i$ when taking into account his randomized action as well.

Next, we will isolate the probability related terms $x_{i}$ on the LHS and all the utility related terms on the RHS and we will integrate over a time interval $[0,T]$ and divide by $T$. I.e.,

\begin{equation}
\label{eq:time_average1}
\frac{\int_0^T\frac{1}{x_{i}} \dot{x}_i dt}{T}= \frac{\int_0^T\big(u_{i0} - \hat{u}_{i}\big)dt}{T}
\end{equation}

\begin{equation}
\label{eq:time_average2}
\frac{\int_0^T-\frac{1}{1-x_{i}} \dot{x}_i dt}{T}= \frac{\int_0^T\big(u_{i1} - \hat{u}_{i}\big)dt}{T}
\end{equation}

However, by a simple change of variables we have that $\int_0^T\frac{1}{x_{i}} \dot{x}_i dt= \ln[x_i(T)]-\ln[x_i(0)]$, 
$\int_0^T-\frac{1}{1-x_{i}} \dot{x}_i dt= -\ln[1-x_i(T)]+\ln[1-x_i(0)]$. The LHS of equations \ref{eq:time_average1}, \ref{eq:time_average2} converge to zero as $T\rightarrow \infty$. Moreover, since the trajectories are periodic (\textit{e.g.}, with period $T_P$), all the following limits exist $\lim_{t\rightarrow \infty}\frac{\int_0^T u_{i0}dt}{T}=\frac{\int_0^{T_P} u_{i0}dt}{T_P}$,
$\lim_{t\rightarrow \infty}\frac{\int_0^T u_{i1}dt}{T}=\frac{\int_0^{T_P} u_{i1}dt}{T_P}$, 
$\lim_{t\rightarrow \infty}\frac{\int_0^T \hat{u}_{i} dt}{T}=\frac{\int_0^{T_P} \hat{u}_{i} dt}{T_P}$.
Thus, for any agent $i$: $\frac{\int_0^{T_P} \hat{u}_{i} dt}{T_P}=\frac{\int_0^{T_P} u_{i0}dt}{T_P}=\frac{\int_0^{T_P} u_{i1}dt}{T_P}$ and since no agent can improve their time average utility over a single period by deviating to any fixed strategy, then, by definition, the time average distribution of play over a single period is a coarse correlated equilibrium of the team zero-sum game.

 Moreover, since in this game any agent has exactly two available strategies any coarse correlated equilibrium is also a correlated equilibrium. This is true, as the only extra allowable deviation that needs to be checked\footnote{Checked in terms of whether it can improve the agent's expected utility.} is the one where he flips his strategy, \textit{i.e.}, whenever he played strategy $1$ he now deviates to strategy $0$, whereas whenever he played $0$ he deviates to $1$. However, we know that each of these deviations cannot improve the agent's expected utility (by the coarse correlated equilibrium property), so the combined deviation does not improve his expected utility either.
\end{proof}

\begin{thm}
\label{thm:main2}
Given any periodic orbit\footnote{We implicitly assume that this corresponds to a generic, safe initial condition.} 
the time average output  $\int_0^T f(t)dt, \int_0^T g(t)dt$ of each team converges to the unique fully mixed Nash equilibrium of the team game, i.e.  $p, q$ respectively. Moreover, the time average expected utility of each agent converges to the value of his team in the $2\times 2$ zero-sum game $U$. 
\end{thm}

\begin{proof}
By lemma \ref{lem:safe_planar}, 
we have that the outputs $f, g$ satisfy a  $(p,q,r,w)$-planar dynamical system:

   \begin{eqnarray*}
   \frac{df}{dt}&=&r(f)(g - q)\\ 
   \frac{dg}{dt}&=&-w(g)(f - p) 
   \end{eqnarray*}

Since $f,g$ are strictly positive along the periodic orbit,\footnote{Since it lies in the interior of $[0,1]^2$ and by the definition of our  $(p,q,r,w)$-planar dynamical systems.} we can isolate $f, g$ on different parts of the equation, and then we integrate over a time interval $[0,T]$ and divide by $T$.

  \begin{eqnarray}
  \frac{\int_0^T \frac{1}{r(f)}\frac{df}{dt} dt}{T}&=& \frac{\int_0^T(g - q)dt}{T}\label{equation:plan1}\\
    \frac{\int_0^T \frac{1}{w(g)}\frac{dg}{dt} dt}{T}&=&-  \frac{\int_0^T(f - p)dt}{T}\label{equation:plan2}
   \end{eqnarray}

However, by a simple change of variables we have that $\int_0^T \frac{1}{r(f)}\frac{df}{dt} dt=\int_0^T \frac{d}{dt}(\int \frac{1}{r(f)}df) dt = \int_{f(0)}^{f(T)}\frac{1}{r(f)}df$. Similarly, $\int_0^T \frac{1}{w(g)}\frac{dg}{dt} dt= \int_{g(0)}^{g(T)}\frac{1}{w(g)}dg$.
Hence, $\lim_{T\rightarrow \infty} \frac{\int_0^T \frac{1}{r(f)}\frac{df}{dt} dt}{T}= \lim_{T\rightarrow \infty} \frac{\int_{f(0)}^{f(T)}\frac{1}{r(f)}df}{T}=0$. Similarly,   $\lim_{T\rightarrow \infty} \frac{\int_0^T \frac{1}{w(g)}\frac{dg}{dt} dt}{T}=0$.
Therefore, $\lim_{T\rightarrow \infty} \frac{\int_0^T g dt}{T}= \frac{\int_0^{T_P} g dt}{T_P}=q$ and similarly,  
$\lim_{T\rightarrow \infty} \frac{\int_0^T f dt}{T}= \frac{\int_0^{T_P} f dt}{T_P}=p$.

\bigskip
Next, we will proceed with the argument about the time average of the agents' utility. 
Let $\vec{f}= \begin{pmatrix} 
f  \\
1-f
\end{pmatrix}$ and $\vec{g}= \begin{pmatrix} 
g  \\
1-g
\end{pmatrix}$
then the expected utility all agents in team $A$ is equal to ${\vec{f}}^T U \vec{g}$.
Let $\vec{p}= \begin{pmatrix} 
p  \\
1-p
\end{pmatrix}$ and $\vec{q}= \begin{pmatrix} 
q  \\
1-q
\end{pmatrix}$ then $({\vec{f}-\vec{p}})^T U (\vec{g}- \vec{q})= {\vec{f}}^T U \vec{g} - {\vec{p}}^T U \vec{q}$, since $\vec{p}, \vec{q}$  are fully mixed Nash equilibrium strategies of the $2\times2$ zero-sum game. So, we have that:

$$\frac{1}{T}\int_0^T({\vec{f}-\vec{p}})^T U (\vec{g}- \vec{q}) dt =\frac{1}{T}\int_0^T {\vec{f}}^T U \vec{g}dt -\frac{1}{T}\int_0^T {\vec{p}}^T U \vec{q}dt = \frac{1}{T}\int_0^T {\vec{f}}^T U \vec{g}dt - {\vec{p}}^T U \vec{q}$$

Therefore in order to argue that $\frac{1}{T}\int_0^T {\vec{f}}^T U \vec{g}dt = {\vec{p}}^T U \vec{q}$, \textit{i.e.}, that the time average of agents' utility is equal to the value of their respective team in the $2\times2$ zero-sum game $U$, it suffices to show that 
$\lim_{T\rightarrow \infty}\frac{1}{T}\int_0^T({\vec{f}-\vec{p}})^T U (\vec{g}- \vec{q}) dt =0$. The payoff matrix $U$ is as follows:

$$
U =
 \begin{pmatrix}
  a & b \\
  c & d
 \end{pmatrix}
 $$

We have that $({\vec{f}-\vec{p}})^T U (\vec{g}- \vec{q})= (a-b-c+d)(f-p)(g-q)$. Therefore it suffices to show that 

 $$\lim_{T\rightarrow \infty}\frac{1}{T}\int_0^T(f-p)(g-q) dt =0.$$
 
 By equation \ref{equation:plan1}, we have already argued that $\lim_{T\rightarrow \infty}\frac{1}{T}\int_0^T(g-q) dt =0$, thus finally we have to show that  $\lim_{T\rightarrow \infty}\frac{1}{T}\int_0^Tf(g-q) dt =0.$ Reworking the equations of the $(p,q,r,w)$-planar dynamical system:

$$ \frac{f}{r(f)} \frac{df}{dt}=f(g - q) \Rightarrow \frac{\int_0^T \frac{f}{r(f)} \frac{df}{dt} dt}{T} =\frac{\int_0^T f(g - q) dt}{T} \Rightarrow \frac{\int_{f(0)}^{f(T)} \frac{f}{r(f)}df}{T}=\frac{\int_0^T f(g - q) dt}{T}.$$

However, $\int_{f(0)}^{f(T)} \frac{f}{r(f)}df$ is bounded and hence   $\lim_{T\rightarrow \infty}\frac{\int_{f(0)}^{f(T)} \frac{f}{r(f)}df}{T}=0$, implying  $\lim_{T\rightarrow \infty}\frac{\int_0^T f(g - q) dt}{T}=0$, and the proof is complete.
\end{proof}

\section{Discussion}
\label{section:Discussion}

We have 
identified a \textit{novel class of conservative dynamical systems} that arise from the simutaneous application of learning dynamics by many (independently acting) agents, provided those agents partition into two teams, with aligned interests within each team and opposed interests across the teams. These learning dynamics are precisely those widely studied recently in learning theory; they are also those studied in the evolutionary theory of sexually reproducing species (``species"=``team of genes"). Beyond having a conserved quantity, these high-dimensional dynamical systems have several properties, each of which was surprising: (a) The dynamics are periodic for almost all initial conditions. (b) The dynamics of the agents not 
  only \textit{minimize regret} (converge to coarse correlated equilibria) but furthermore the time average of any trajectory \textit{converges point-wise to a correlated equilibrium} that is not necessarily a Nash equilibrium.
  (c) The time average play of these trajectories are shown to \textit{implement the minmax strategies of each species} in its struggle for survival against the opposing species.  
  (d) The time average utility of each agent, under these dynamics, is actually the same as if each team was cooperating to play its best strategy at every moment in time---despite the fact that the agents are not cooperating (the distribution of play by a team is always a product distribution across agents), and that their play is not constant in time nor converging to a constant in time.

Each of these properties is, to our knowledge, without precedent in high dimensional dynamics of learning / sexual evolution. 
Cyclic trajectories are known to exist in low dimensional game theoretic systems, systems of species competition or even hypercycle equations, but no such result is known  for larger systems  \cite{Hofbauer98}. Conversely, there exist many simple low dimensional examples where replicator dynamics can have  complex trajectories \cite{Sandholm10}.

The results of this paper are obtained entirely in an infinite-population, continuous-time limit. Finite-population, discrete-time models have a number of defects, not least of which are that they have too many adjustable parameters, and that rounding and drift terms turn conserved quantities into ``almost conserved" quantities. Having said this, the issue cannot be ignored, as there are known situations in which finite-population models have qualitatively different behavior than their corresponding infinite-population model~\cite{Arora:1994:SQD:195058.195231}. In our situation, the periodicity theorem actually guarantees that the dynamics of the system keep every allele frequency bounded away from $0$ and $1$; in the short run this precludes irreversible rounding errors that occur when such a frequency is rounded to $0$ or $1$. Of course, roundings by $1/\text{population}$ must occur, so in the long run, such rounding may accumulate until the dynamics are very far from the infinite-limit cycle, and then, eventually, roundings to $0$ or $1$ can occur. For moderate time scales we do not expect this to affect the predictions of the paper. (For long time scales one must keep in mind that our theorems depend also on other idealizations of the model, notably the weak selection hypothesis and the absence of other effects such as mutations or selective mating patterns.)

\bibliographystyle{plain}
\bibliography{sigproc4}

\begin{thebibliography}{10}

\bibitem{Akin84}
E.~Akin and V.~Losert.
\newblock Evolutionary dynamics of zero-sum games.
\newblock {\em J. of Math. Biology}, 20:231--258, 1984.

\bibitem{Arnold78}
V.~I. Arnold.
\newblock {\em Ordinary Differential Equations}.
\newblock The MIT Press, 1978.

\bibitem{Arora05themultiplicative}
S.~Arora, E.~Hazan, and S.~Kale.
\newblock The multiplicative weights update method: a meta algorithm and
  applications.
\newblock Technical report, 2005.

\bibitem{Arora:1994:SQD:195058.195231}
S.~Arora, Y.~Rabani, and U.~Vazirani.
\newblock Simulating quadratic dynamical systems is {PSPACE}-complete
  (preliminary version).
\newblock In {\em Proceedings of the Twenty-sixth Annual ACM Symposium on
  Theory of Computing}, STOC, pages 459--467. ACM, 1994.

\bibitem{aumann1974subjectivity}
R.~J. Aumann.
\newblock Subjectivity and correlation in randomized strategies.
\newblock {\em Journal of mathematical Economics}, 1(1):67--96, 1974.

\bibitem{Bell82}
G.~Bell.
\newblock {\em The Masterpiece Of Nature: The Evolution and Genetics of
  Sexuality}.
\newblock U California Press, 1982.

\bibitem{bendixson1901courbes}
I.~Bendixson.
\newblock Sur les courbes d{\'e}finies par des {\'e}quations
  diff{\'e}rentielles.
\newblock {\em Acta Mathematica}, 24(1):1--88, 1901.

\bibitem{doi:10.1086/376580}
C.~W. Benkman, T.~L. Parchman, A.~Favis, and A.~M. Siepielski.
\newblock Reciprocal selection causes a coevolutionary arms race between
  crossbills and lodgepole pine.
\newblock {\em The American Naturalist}, 162(2):182--194, 2003.

\bibitem{pota}
A.~Blum, M.~Hajiaghayi, K.~Ligett, and A.~Roth.
\newblock Regret minimization and the price of total anarchy.
\newblock In {\em Proceedings of the 40th annual ACM symposium on Theory of
  computing}, STOC, pages 373--382, 2008.

\bibitem{brodie2002the}
E.~D. {Brodie Jr.}, B.~J. Ridenhour, and E.~D.~Brodie III.
\newblock The evolutionary response of predators to dangerous prey: Hotspots
  and coldspots in the geographic mosaic of coevolution between garter snakes
  and newts.
\newblock {\em Evolution}, 56(10):2067--2082, 10 2002.

\bibitem{Chastain:2013:MUC:2422436.2422444}
E.~Chastain, A.~Livnat, C.~Papadimitriou, and U.~Vazirani.
\newblock Multiplicative updates in coordination games and the theory of
  evolution.
\newblock In {\em Proceedings of the 4th Conference on Innovations in
  Theoretical Computer Science}, ITCS '13, pages 57--58, New York, NY, USA,
  2013. ACM.

\bibitem{PNAS2:Chastain16062014}
E.~Chastain, A.~Livnat, C.~Papadimitriou, and U.~Vazirani.
\newblock Algorithms, games, and evolution.
\newblock {\em Proceedings of the National Academy of Sciences}, 2014.

\bibitem{NonConvergingDask}
C.~Daskalakis, R.~Frongillo, C.~H. Papadimitriou, G.~Pierrakos, and G.~Valiant.
\newblock On learning algorithms for {N}ash equilibria.
\newblock In S.~Kontogiannis, E.~Koutsoupias, and P.~G. Spirakis, editors, {\em
  Algorithmic Game Theory}, volume 6386 of {\em Lecture Notes in Computer
  Science}, pages 114--125. Springer Berlin Heidelberg, 2010.

\bibitem{Dawkins1976}
R.~Dawkins.
\newblock {\em The Selfish Gene}.
\newblock Oxford University Press, 1976.

\bibitem{EhrlichR64}
P.~R. Ehrlich and P.~H. Raven.
\newblock Butterflies and plants: A study in coevolution.
\newblock {\em Evolution}, 18(4):586--608, Dec. 1964.

\bibitem{Eigen71}
M.~Eigen.
\newblock Selforganization of matter and the evolution of biological
  macromolecules.
\newblock {\em Naturwissenschaften}, 58(10):465--523, 1971.

\bibitem{EigenS79}
M.~Eigen and P.~Schuster.
\newblock {\em The Hypercycle: A Principle of Natural Self-Organization}.
\newblock Springer-Verlag, 1979.

\bibitem{fuks1984beginner}
D.~B. Fuks, V.~A. Rokhlin, and A.~Iacob.
\newblock {\em Beginner's course in topology: geometric chapters}.
\newblock Springer, 1984.

\bibitem{mathbio1}
S.~Gavrilets.
\newblock {\em Fitness Landscapes and the Origin of Species}.
\newblock Princeton University Press, 2004.

\bibitem{Hofbauer98}
J.~Hofbauer and K.~Sigmund.
\newblock {\em Evolutionary Games and Population Dynamics}.
\newblock Cambridge University Press, Cambridge, 1998.

\bibitem{paperics11}
R.~Kleinberg, K.~Ligett, G.~Piliouras, and {\'E}.~Tardos.
\newblock Beyond the {Nash} equilibrium barrier.
\newblock In {\em Symposium on Innovations in Computer Science (ICS)}, 2011.

\bibitem{Kleinberg09multiplicativeupdates}
R.~Kleinberg, G.~Piliouras, and {\'E}.~Tardos.
\newblock Multiplicative updates outperform generic no-regret learning in
  congestion games.
\newblock In {\em ACM Symposium on Theory of Computing (STOC)}, 2009.

\bibitem{sigecom11}
K.~Ligett and G.~Piliouras.
\newblock Beating the best \text{Nash} without regret.
\newblock {\em SIGecom Exchanges 10}, 2011.

\bibitem{CACM}
A.~Livnat and C.~Papadimitriou.
\newblock Sex as an algorithm: the theory of evolution under the lens of
  computation.
\newblock {\em Communications of the ACM (CACM)}, 59:84--93, November 2016.

\bibitem{evolfocs14}
A.~Livnat, C.~Papadimitriou, A.~Rubinstein, A.~Wan, and G.~Valiant.
\newblock Satisfiability and evolution.
\newblock In {\em FOCS}, 2014.

\bibitem{Marschak55}
J.~Marschak.
\newblock Elements for a theory of teams.
\newblock {\em Management Science}, 1(2):127--137, 1955.

\bibitem{ITCS15MPP}
R.~Mehta, I.~Panageas, and G.~Piliouras.
\newblock Natural selection as an inhibitor of genetic diversity.
\newblock In {\em ITCS}, 2015.

\bibitem{ITCS17MPPPV}
R.~Mehta, I.~Panageas, G.~Piliouras, P.~Tetali, and V.~V. Vazirani.
\newblock Mutation, sexual reproduction and survival in dynamic environments.
\newblock In {\em ITCS}, 2017.

\bibitem{mehta_et_al:LIPIcs:2016:6407}
R.~Mehta, I.~Panageas, G.~Piliouras, and S.~Yazdanbod.
\newblock {The Computational Complexity of Genetic Diversity}.
\newblock In P.~Sankowski and C.~Zaroliagis, editors, {\em 24th Annual European
  Symposium on Algorithms (ESA 2016)}, volume~57 of {\em Leibniz International
  Proceedings in Informatics (LIPIcs)}, pages 65:1--65:17, Dagstuhl, Germany,
  2016. Schloss Dagstuhl--Leibniz-Zentrum fuer Informatik.

\bibitem{CyclesSODA18}
P.~{Mertikopoulos}, C.~{Papadimitriou}, and G.~{Piliouras}.
\newblock {Cycles in adversarial regularized learning}.
\newblock In {\em ACM-SIAM Symposium on Discrete Algorithms (SODA)}, 2018.

\bibitem{Nagylaki1}
T.~Nagylaki.
\newblock The evolution of multilocus systems under weak selection.
\newblock {\em Genetics}, 134(2):627--47, 1993.

\bibitem{NowakO08}
M.~A. Nowak and H.~Ohtsuki.
\newblock Prevolutionary dynamics and the origin of evolution.
\newblock {\em Proceedings of the National Academy of Sciences},
  105(39):14924--14927, 2008.

\bibitem{ChaosNIPS17}
G.~{Palaiopanos}, I.~{Panageas}, and G.~{Piliouras}.
\newblock {Multiplicative Weights Update with Constant Step-Size in Congestion
  Games: Convergence, Limit Cycles and Chaos}.
\newblock In {\em NIPS}, 2017.

\bibitem{CRS16}
Christos Papadimitriou and Georgios Piliouras.
\newblock From nash equilibria to chain recurrent sets: Solution concepts and
  topology.
\newblock In {\em ITCS}, 2016.

\bibitem{10.2307/3298524}
O.~Pellmyr.
\newblock Yuccas, yucca moths, and coevolution: A review.
\newblock {\em Annals of the Missouri Botanical Garden}, 90(1):35--55, 2003.

\bibitem{PiliourasAAMAS2014}
G.~Piliouras, C.~Nieto-Granda, H.~I. Christensen, and J.~S. Shamma.
\newblock Persistent patterns: Multi-agent learning beyond equilibrium and
  utility.
\newblock In {\em Proceedings of the 2014 International Conference on
  Autonomous Agents and Multi-agent Systems}, AAMAS '14, pages 181--188,
  Richland, SC, 2014. International Foundation for Autonomous Agents and
  Multiagent Systems.

\bibitem{Soda14}
G.~Piliouras and J.~S. Shamma.
\newblock Optimization despite chaos: Convex relaxations to complex limit sets
  via {P}oincar\'{e} recurrence.
\newblock In {\em SODA}, 2014.

\bibitem{Sandholm10}
W.~H. Sandholm.
\newblock {\em Population Games and Evolutionary Dynamics}.
\newblock MIT Press, 2010.

\bibitem{Sato02042002}
Y.~Sato, E.~Akiyama, and J.~D. Farmer.
\newblock Chaos in learning a simple two-person game.
\newblock {\em Proceedings of the National Academy of Sciences},
  99(7):4748--4751, 2002.

\bibitem{SchV17}
L.~J. Schulman and U.~V. Vazirani.
\newblock The duality gap for two-team zero-sum games.
\newblock In {\em Proc.~ITCS}, 2017.

\bibitem{Schuster1983533}
P.~Schuster and K.~Sigmund.
\newblock Replicator dynamics.
\newblock {\em Journal of Theoretical Biology}, 100(3):533 -- 538, 1983.

\bibitem{Taylor1978145}
P.~D. Taylor and L.~B. Jonker.
\newblock Evolutionary stable strategies and game dynamics.
\newblock {\em Mathematical Biosciences}, 40(12):145--156, 1978.

\bibitem{teschl2012ordinary}
G.~Teschl.
\newblock {\em Ordinary differential equations and dynamical systems}, volume
  140.
\newblock American Mathematical Soc., 2012.

\bibitem{Thompson94}
J.~N. Thompson.
\newblock {\em The Coevolutionary Process}.
\newblock U Chicago Press, 1994.

\bibitem{Thompson05}
J.~N. Thompson.
\newblock {\em The Geographic Mosaic of Coevolution}.
\newblock U Chicago Press, 2005.

\bibitem{ThompsonC02}
J.~N. Thompson and B.~M. Cunningham.
\newblock Geographic structure and dynamics of coevolutionary selection.
\newblock {\em Nature}, 417:735--738, 2002.

\bibitem{valen73}
L.~Van Valen.
\newblock A new evolutionary law.
\newblock {\em Evolutionary Theory}, 1:1--30, 1973.

\bibitem{Valiant-book}
L.~Valiant.
\newblock {\em Probably Approximately Correct: Nature's Algorithms for Learning
  and Prospering in a Complex World}.
\newblock Basic Books, 2013.

\bibitem{VONSTENGEL1997309}
B.~von Stengel and D.~Koller.
\newblock Team-maxmin equilibria.
\newblock {\em Games and Economic Behavior}, 21(1):309 -- 321, 1997.

\bibitem{Weibull}
J.~W. Weibull.
\newblock {\em Evolutionary Game Theory}.
\newblock MIT Press; Cambridge, MA, 1995.

\bibitem{young2004strategic}
H.~P. Young.
\newblock {\em Strategic Learning and Its Limits}.
\newblock Arne Ryde memorial lectures. Oxford University Press, 2004.

\end{thebibliography}

\appendix

\section{Proof of Lemma \ref{lem:sys}}
\label{Appendix:ProofSys}

\begin{proof}
By assumption  the $2\times2$ zero sum game  has a unique Nash equilibrium with $p=\frac{d-c}{a-b-c+d}, q=\frac{d-b}{a-b-c+d}$. By substitution in equation \ref{eq:repli_A} we derive:

\begin{eqnarray*}
\dot{x}_i &=& x_{i}(1-x_{i})\big(u_{i0} - u_{i1}\big)\\
              &=& x_{i}(1-x_{i})(f_{i0}(ga+(1-g)b)+(1-f_{i0})(gc+(1-g)d)-\\
              &-&  f_{i1}(ga+(1-g)b)-(1-f_{i1})(gc+(1-g)d))\\
              &=& x_{i}(1-x_{i})( f_{i0}-f_{i1})\big(g(a-b-c+d)-(d-b)\big) \\
              &=& (a-b-c+d)x_{i}(1-x_{i})( f_{i0}-f_{i1})\big(g-\frac{d-b}{a-b-c+d}\big) \\ 
              &=&  \alpha x_i(1-x_i)(f_{i0}-f_{i1})(g-q)
\end{eqnarray*}
\noindent
where $\alpha =  (a-b-c+d)\neq 0$ by assumption. By substitution in equation \ref{eq:repli_B} we derive the analogous system for team $B$. 
\end{proof}

\section{Proof of Theorem \ref{thm:safe}}
\label{Appendix:ProofSafe}


The high level idea is that
by definition, if $\xx_0$ and $\yy_0$ are safe for subsystems  $\Phi^A$ and $\Phi^B$   respectively then $(\xx_0,\yy_0)$ is safe for system $\Phi$ as well.
However,  subsystems $\Phi^A, \Phi^B$  correspond to replicator dynamics being  applied to a potential game. By applying a game theoretic characterization of stable equilibria in potential games for replicator dynamics developed in  \cite{Kleinberg09multiplicativeupdates} we can show that replicator dynamics does not converge to such randomized states for all but a zero measure of initial conditions in each subsystem. Thus, all but a measure 0 set of  initial conditions are safe in $\Phi$ as well. 

\noindent
We will start the formal proof by providing the relevant definitions and theorems from \cite{Kleinberg09multiplicativeupdates}. 

\begin{definition}[\cite{Kleinberg09multiplicativeupdates}]
\label{def:weakly}
A Nash equilibrium is called {\bf weakly stable} if each agent $i$ remains indifferent between the strategies in the support of his (mixed) strategy $\xx_i$ whenever any other single player $j$ modifies his mixed strategy to any pure strategy in the support of his (mixed) strategy $\xx_j$.
\end{definition}

\begin{thm}[\cite{Kleinberg09multiplicativeupdates}]
\label{thm:Klein}
For all initial conditions replicator dynamics converges to equilibria in potential games.
If  a fixed point that is not a Nash equilibrium, the set of initial conditions converging to is of zero measure.
If a fixed point is not a weakly stable Nash equilibrium, the set of initial conditions converging to is of zero measure.
\end{thm}

\begin{lem}
\label{lem:weakly}
If $\xx$ is an isolated, not pure, fixed  point of subsystem $A$, \textit{i.e.},  with at least one randomizing agent $i$ such that $f_{i0}=f_{i1}=f$,   
then it cannot be a weakly stable Nash equilibrium in the partnership/potential game where all agents receive in each outcome utility equal to the output of their Boolean function $f(\xx)$. A symmetric statement holds for $\yy$, isolated, not pure, fixed  point of subsystem $B$.
\end{lem}

\begin{proof}
We will prove the contrapositive.
If an equilibrium is weakly stable that means by definition \ref{def:weakly} that
given any single deviation of a randomizing agent $i$ to a pure strategy no other randomizing agent can deviate and improve his/her payoff.
However, this means that in both outcomes $(0,${\boldmath $x$}$_{-i})$, $(1,${\boldmath $x$}$_{-i})$ (for which $f_{i0}=f_{i1}=f$ by the  fixed point property) 
no  randomizing agent can deviate and increase his/her payoff. This implies that $f_{i0i'0}=f_{i0i'1}=f_{i0}=f$ for any randomizing agent $i'$ in  $\xx$ (of team $A$). 
Similarly, that $f_{i1i'0}=f_{i1i'1}=f_{i1}=f$ for any  randomizing agent $i'$  in team $A$.
Hence the outcomes $(0,${\boldmath $x$}$_{-i})$, $(1,${\boldmath $x$}$_{-i})$ are fixed points of subsystem $A$.
 Moreover, as long as exactly two (randomizing) agents, 
  fix their behavior then the expected output of the Boolean function does not change. 
 We will show that any product distribution $(\delta,${\boldmath $x$}$_{-i})$
where $0<\delta<1$  is also a fixed point. It suffices to check all the (randomizing) agents other than $i$.
If this set is empty we are done. Otherwise,
 let $i'$ be such an agent.
Any such agent when deviating to a pure strategy receives utility equal to the expected output of team's $A$ Boolean function.
For example if he deviates to his first strategy then he receives $f_{i\delta i'0}=\delta f_{i0i'0}+(1-\delta)f_{i1i'0}=f$.
Hence,   $(\delta,${\boldmath $x$}$_{-i})$ is a fixed point as well for any $\delta\in [0,1]$ and the initial  fixed point {\boldmath $x$} is not isolated.
\end{proof}

\noindent
We are now ready to complete the proof of Theorem \ref{thm:safe}.

\begin{proof}
 By definition,  if $\xx_0$ and $\yy_0$ are safe for subsystems  $\Phi^A$ and $\Phi^B$   respectively then $(\xx_0,\yy_0)$ is safe for system $\Phi$ as well.
However,  subsystems $\Phi^A, \Phi^B$  correspond to replicator dynamics being  applied to a potential game. 
By lemma \ref{lem:weakly} we have that an isolated, not pure, fixed  point of subsystem $A$ or $B$,
cannot be a weakly stable Nash equilibrium in the partnership/potential game where all agents receive in each outcome utility equal to the output of their Boolean function $f(\xx)$. By applying Theorem \ref{thm:Klein} the set of initial conditions that converge to any such equilibrium as $t\rightarrow \infty$ is a measure zero set. By applying Theorem \ref{thm:Klein} on the cost minimization potential game\footnote{A fully mixed fixed point (or more generally a fixed point where for each agent $i$ his expected utility remains constant when unilaterally deviating to any other strategy) is a Nash equilibrium in both the payoff maximization and the cost minimization game, \textit{i.e.}, it does not matter if you receive or pay $f(x)$ and in both games for the same reasons it is not a weakly mixed Nash. On the other hand, if there exists an agent $i$ such that when unilaterally  deviating to another strategy his payoff decreases then in the cost game this fixed point is no longer a Nash equilibrium and Theorem  \ref{thm:Klein} still implies a zero measure of initial conditions converging to it.} where at any outcome $\xx$ agents have to pay cost equal to the output of their Boolean function $f(\xx)$, we derive that the set of initial conditions that converge in subsystem $A$ to any such equilibrium as $t\rightarrow -\infty$ is still a measure zero set. Finally, by assumption all system equilibria are isolated and thus there can only be finitely many of them.\footnote{Otherwise, given any countable set of equilibria by compactness of the state space there will be a concentration point in this sequence, which leads to a contradiction since every fixed point is isolated.}  So, by taking union bound over all not pure, isolated fixed points  we have that the set of all initial conditions converging to them as $t\rightarrow \pm\infty$ in each subsystem is a zero measure set. Since the state space $\Phi$ is the product of the state spaces of $\Phi^A, \Phi^B$, the set of all such conditions $(\xx_0,\yy_0)$ such that either $\xx_0$ or $\yy_0$ are not safe is a measure zero set of the state space of $\Phi$. 
\end{proof}

\section{Background on Topology, Manifolds, and Dynamics}
\label{Section:Background}

\begin{definition}
Let $U \subset \Real^n$ and $V \subset \Real^k$ be open sets. A mapping $f$ from
$U$ to $V$ written as $f : U \rightarrow V$ is a smooth map if all of the partial derivatives
exist and are continuous.
\end{definition}

\noindent
More generally we have the following:

\begin{definition}
Let $X \subset \Real^n$ and $Y \subset \Real^k$ be arbitrary subsets, and a map
$f : X \rightarrow Y$ is called smooth map if for every $x \in X$ there exist an open set
$O \subset R^n$ with $x \in O$ and a smooth mapping $F : O \rightarrow \Real^k$ that coincides with
$f$ throughout $O \cap X$.
\end{definition}


A $1$-manifold (or manifold of dimension $1$) is a topological space which is second countable (\textit{i.e.}, its topological structure has a countable base), 
satisfies the Hausdorff axiom (\textit{i.e.}, any two different points have disjoint neighborhoods) and each point of which has a neighborhood homeomorphic either to the real line $\Real$ or to the half-line $\Real_+=\{x\in \Real:x\geq 0\}$.

\begin{thm}[\cite{fuks1984beginner}]
Any connected closed $1$-manifold is homeomorphic to $S^1= \{(x,y)\in \Real^2: x^2+y^2=1\}$.
\end{thm}

\begin{definition}
Let $U, V$ be manifolds.
A map $f : U \rightarrow V$ is called a diffeomorphism if $f$ carries
$U$ onto $V$ and also both $f$ and $f^{-1}$ are smooth.
\end{definition}

\begin{definition}
Let $f : U \rightarrow V$ be a smooth map between same dimensional manifolds. We denote that $x \in U$ is a regular point if the derivative is nonsingular.
$y \in V$ is called a {\bf regular value} if $f^{-1}(y)$ contains only regular points.
If the derivative is singular, then $x$ is called a {\bf critical point}. We also say $y \in V$ is a critical value if $y$ is not a regular value.
\end{definition}

For a smooth function on $\Real^n$ to $\Real$, a point  $p$  is critical if all of the partial derivatives of the function are zero at $p$, or, equivalently, if its gradient is zero. Given a differentiable map $f$ from $R^m$ into $R^n$, the critical points of $f$ are the points of $R^m$, where the rank of the Jacobian matrix of $f$ is not maximal. 

\begin{thm}[Sard's Theorem] 
Let $f : X \rightarrow Y$ be a smooth map of manifolds, and let $C$ be the set of critical points of $f$ in $X$. Then $f(C)$ has measure zero in $Y$.
\end{thm}

\begin{thm}[Regular Value Theorem]
 If $y \in Y$ is a regular value of $f : X \rightarrow Y$ then $f^{-1}(y)$ is a manifold of dimension $n-m$, since $dim(X) = n$ and 
 $dim(Y ) = m$.
\end{thm}

\begin{thm}[Inverse Function Theorem]
Let $U, V$ be open sets of $R^n$.
If $f : U \rightarrow V$ is a smooth map and at a point p the jacobian matrix $df_p$ is
invertible, then there is a neighborhood $U$ of $p$ on which $f : U \rightarrow f(U)$ is a
diffeomorphism.
\end{thm}

For functions of a single variable, the theorem states that if $f$ is a smooth function with nonzero derivative at the point $p$, then $f$ is invertible in a neighborhood of $p$, the inverse is smooth, and $\bigl(f^{-1}\bigr)'(f(p)) = \frac{1}{f'(p)},$ where the left side notation refers to the derivative of the inverse function evaluated at $f(p)$.

\subsection*{Poincar\'{e}-Bendixson Theorem}

The Poincar\'{e}-Bendixson theorem is a powerful theorem that implies that two-dimensional systems cannot  exhibit chaos.
Effectively, the limit behavior is either going to be an equilibrium, a periodic orbit, or a closed loop, punctuated by one (or more) fixed points.
Formally, we have:

\begin{thm}
[\cite{bendixson1901courbes,teschl2012ordinary}]
Given a differentiable real dynamical system defined on an open subset of the plane, then every non-empty compact $\omega$-limit set of an orbit, which contains only finitely many fixed points, is either
a fixed point, a periodic orbit, or a connected set composed of a finite number of fixed points together with homoclinic and heteroclinic orbits connecting these.
\end{thm}

 
 
\end{document}